\documentclass[journal,onecolumn,12pt]{IEEEtran}
\usepackage[numbers,sort,compress]{natbib}
\usepackage{pslatex}
\usepackage{amsfonts,color,morefloats}
\usepackage{amssymb,amsmath,latexsym,amsthm,mathrsfs}
\usepackage{hyperref}
\usepackage{xcolor}
\usepackage{makecell}
\usepackage{bm}
\usepackage{xfrac}

\usepackage[utf8]{inputenc}
\usepackage[T1]{fontenc}  

\hypersetup{hidelinks}
\allowdisplaybreaks

\theoremstyle{definition}

\newtheorem{thm}{Theorem}[section]
\newtheorem{lem}[thm]{Lemma}

\newcommand{\wt}{{\mathrm{wt}}}

\newcommand{\gf}{{\mathbb{F}}}
\newcommand{\F}{{\mathbb{F}}}

\newcommand{\N}{\mathbb{N}}

\newcommand{\mb}[1]{\mathbf{#1}}

\newcommand{\supp}{\mathrm{supp}}
\newcommand{\spa}{\mathrm{span}}
\newcommand{\BCH}{\mathrm{BCH}}
\newcommand{\tr}{\mathrm{Tr}}

\newcommand{\sqbinom}[2]{\genfrac{[}{]}{0pt}{}{#1}{#2}}

\newcommand{\colvec}[2]{%
	\left[\begin{matrix}
		#1 \\
		#2
	\end{matrix}\right]}

\begin{document}
\title{On generalized covering radii of binary primitive double-error-correcting BCH codes}
\author{Maosheng~Xiong and~Chi~Hoi~Yip\\
\thanks{M.~Xiong is with the Department of Mathematics, The Hong Kong University of Science and Technology, Hong Kong, P. R. China. E-mail: \href{mailto: mamsxiong@ust.hk}{mamsxiong}@ust.hk}
\thanks{C. H.~Yip is with the Department of Mathematics, The Hong Kong University of Science and Technology, Hong Kong, P. R. China. E-mail: \href{mailto: machyip@ust.hk}{machyip}@ust.hk}}

\maketitle
\begin{abstract}
	The generalized covering radii (GCR) of linear codes are a fundamental higher-dimensional extension of the classical covering radius. While the second and third GCR of binary primitive double-error-correcting BCH codes, \(\text{BCH}(2,m)\), were recently determined, their proofs relied on highly complex combinatorial arguments, and the behavior of the GCR hierarchy for larger orders \(t\) has remained largely unexplored. In this paper, we introduce the Generalized Supercode Lemma, which lower-bounds the GCR of a code using the generalized Hamming weights of an appropriate supercode. Applying this lemma, we significantly streamline and simplify the proofs for the known lower bounds of \(\rho_2(\text{BCH}(2,m))\) and \(\rho_3(\text{BCH}(2,m))\), and we establish a new lower bound for \(\rho_4(\text{BCH}(2,m))\). Furthermore, by combining combinatorial arguments with Weil-type exponential sum estimates, we investigate the GCR hierarchy for general \(t\), proving that \(2t \le \rho_t(\text{BCH}(2,m)) \le 2t+1\) whenever \(m\) is sufficiently large compared to \(t\).
\end{abstract}

	

\thispagestyle{empty}   

\section{Introduction} \label{sec:intro}

The covering radius~\(\rho(C)\) of a linear code~\(C\subseteq \F_q^n\) measures how well the codewords fill the ambient Hamming space: it is the maximum Hamming distance from an arbitrary vector in \(\F_q^n\) to the nearest codeword. Together with length, dimension, and minimum distance~\(d(C)\), it is one of the four fundamental parameters of a code~\cite{Delsarted,Cohenas}.

Motivated by applications in database linear querying (including private information retrieval), Elimelech, Firer, and Schwartz~\cite{Elimelechl} introduced a higher-dimensional extension of \(\rho(C)\), the \emph{generalized covering radii} (GCR) \(\rho_t(C)\) for \(t \ge 1\), where \(\rho_1(C)=\rho(C)\). Informally, \(\rho_t(C)\) is the smallest integer \(r\) such that any \(t\) target vectors can be brought within distance \(r\) of $t$ codewords using a common set of coordinate positions; equivalently, any \(t\)-tuple in \((\F_q^n)^t\) can be \emph{simultaneously} covered by a \(t\)-dimensional Hamming ball centered at a \(t\)-tuple of codewords. The GCR arise naturally in joint recovery of linear computations over finite fields~\cite{Elimelechl} and in recent low-access and quantized linear-computation settings~\cite{Ramkumar}, where they capture trade-offs between access complexity, latency, and redundancy.

The study of GCR is relatively recent and limited. Beyond the foundational work~\cite{Elimelechl}, current results include GCR bounds for Reed--Muller codes~\cite{Elimelechn,Langton} and asymptotic rate bounds for binary codes with prescribed normalized GCR~\cite{Elimelechl,Elimelechh}. Recently, the second and third GCR of binary primitive double-error-correcting BCH codes, denoted \(\BCH(2,m)\), were determined in~\cite{Yohananovb,Ozbudakh}. 

Because \(\BCH(2,m)\) is among the most important classes of binary cyclic codes of length \(n=2^m-1\) with a rich algebraic structure and broad use in communication systems, it serves as a natural testbed for understanding GCR \(\rho_t(\BCH(2,m))\) beyond small orders of $t$. However, existing analyses for $t=2$ and $t=3$ carried out in \cite{Yohananovb,Ozbudakh} rely on highly sophisticated and complex arguments, and the behavior of the GCR hierarchy for larger $t$ has remained largely unexplored. 

In this paper, we study the full GCR hierarchy \(\rho_t(\BCH(2,m))\). Our approach not only establishes bounds for arbitrary $t$ but also significantly simplifies the existing analyses for small $t$. Specifically, our main contributions are twofold:

\begin{itemize}
    \item \textbf{The Generalized Supercode Lemma:} We establish a general lower bound for GCR that parallels the classical supercode bound for $\rho(C)$. We show that if $C \subset C'$ and $\dim C'-\dim C \ge t$, then $\rho_t(C)$ can be lower-bounded by the $t$-th generalized Hamming weight of the supercode $C'$. Armed with a refined version of this lemma and a coding-theoretic argument, we recover the known lower bounds for $\rho_2(\BCH(2,m))$ \cite{Yohananovb} and $\rho_3(\BCH(2,m))$ \cite{Ozbudakh} in a much simpler and more transparent way. We also obtain new results on $\rho_4(\BCH(2,m))$ in this way. 
    
    \item \textbf{General Bounds for Large \(m\):} We extend the study of the GCR hierarchy to general orders. By utilizing combinatorial arguments alongside standard number-theoretical tools such as character sums and Weil-type estimates, we prove that \(2t \le \rho_t(\BCH(2,m)) \le 2t+1\) for any \(t\) when \(m\) is sufficiently large compared with \(t\). 
\end{itemize}

The remainder of the paper is organized as follows. Section~\ref{sec:pre} introduces notation and recalls generalized covering radii, generalized Hamming weights, primitive BCH codes, and the character-sum estimates needed later. In Section~\ref{sec:rho3-rho4}, we prove the Generalized Supercode Lemma and apply it to study \(\rho_2, \rho_3\) and \(\rho_4\) of \(\BCH(2,m)\). Section~\ref{sec:general-bounds} establishes general lower and upper bounds for \(\rho_t(\mathrm{BCH}(2,m))\) when \(m\) is sufficiently large. An appendix contains a proof of the uniqueness (up to permutation equivalence) of the binary \([6,3,3]\) linear code, used in the argument for \(\rho_3(\BCH(2,m))\) is given in Section~\ref{sec:appenx}. Then in Section~\ref{sec:cond} we conclude this paper. 
 
\section{Preliminaries} \label{sec:pre} 

In this section, we introduce the fundamental concepts and notations used throughout the paper. We start with the basic parameters of linear codes, introduce their generalized counterparts, and finally review primitive BCH codes and Weil-type exponential sum estimates.

\subsection{Notation} 

We borrow some notation from \cite{Elimelechl}. 

For any $n \in \N$, we define $[n]\triangleq \{1,2,\ldots,n\}$. For a finite set $A$ and $k \in \N$, we denote by $\binom{A}{k}$ the set of all subsets of $A$ of size exactly $k$. We denote by $|A|$ the cardinality of $A$. We use $\F_q$ to denote the finite field of size $q$, where $q$ is a prime power, and denote $\F_q^* \triangleq \F_q \setminus \{0\}$. For a vector space $V$ over $\F_q$ and $k \in \N$, we denote by $\sqbinom{V}{k}$ the set of all vector subspaces of $V$ of dimension exactly $k$. 

For a matrix $H$ over $\F_q$ with $n$ columns, we denote by $\mb{h}_i$ its $i$-th column. For any subset $I=\{i_1,i_2,\ldots,i_t\} \in \binom{[n]}{t}$, we denote $H_I \triangleq \left[\mb{h}_{i_1},\ldots,\mb{h}_{i_t}\right]$, i.e., $H_I$ is the matrix formed from $H$ by choosing the columns whose indices are from $I$. We use $\langle H_I\rangle$ to denote the vector space spanned by the columns of $H_I$, i.e.,
\[\langle H_I \rangle \triangleq \mathrm{span}_{\F_q} \, \left\{ \mb{h}_{i_1}, \mb{h}_{i_2},\ldots,\mb{h}_{i_t} \right\}.\]

Given $\mb{v}=(v_1,\ldots,v_n) \in \F_q^n$, the support of $\mb{v}$ is defined by
\[\supp (\mb{v}) \triangleq \left\{i \in [n]: v_i \ne 0\right\}. \]
For a subset $V \subset \F_q^n$ we define 
\[\supp(V) \triangleq \bigcup_{\mb{v} \in V} \supp(\mb{v}). \]
The Hamming weight of $\mb{v}$ is then defined as $\wt(\mb{v}) \triangleq \left|\supp(\mb{v})\right|$. The Hamming distance between two vectors $\mb{v}$ and $\mb{v}'$ of $\F_q^n$ is given by $d(\mb{v},\mb{v}')\triangleq \wt(\mb{v}-\mb{v}')$, i.e., the Hamming distance between $\mb{v}$ and $\mb{v}'$ is the number of coordinate positions in which they differ. 

\subsection{Fundamental Parameters of Linear Codes}

A linear code $C$ of length $n$ over $\mathbb{F}_q$ is simply a vector subspace of the ambient space $\mathbb{F}_q^n$. The elements of $C$ are called codewords. The \emph{minimum distance} $d(C)$ of a code $C$ is the smallest Hamming distance between any two distinct codewords, or equivalently, the minimum Hamming weight of all nonzero codewords of $C$: 
\[
d(C) \triangleq \min_{\substack{\mb{c}_1, \mb{c}_2 \in C \\ \mb{c}_1 \neq \mb{c}_2}} d(\mb{c}_1, \mb{c}_2) = \min_{\substack{\mb{c} \in C \\ \mb{c} \neq \mb{0}}} \text{wt}(\mb{c}).
\]
The minimum distance dictates the error-detecting and error-correcting capability of the code. We say $C$ is an $[n,k,d]_q$ code if $C$ is a linear code of length $n$ over $\F_q$ with dimension $k$ and minimum distance at least $d$. 

The \emph{covering radius} $\rho(C)$ of a linear code $C$ measures how well the code covers the entire ambient space $\mathbb{F}_q^n$. It is defined as the smallest integer $r$ such that every vector in $\mathbb{F}_q^n$ lies within a Hamming ball of radius $r$ centered at some codeword of $C$. Formally:
\[
\rho(C) \triangleq \max_{\mb{v} \in \mathbb{F}_q^n} \,\,\min_{\mb{c} \in C} \,d(\mb{v}, \mb{c}).
\]

Another way to describe the covering radius $\rho(C)$ is as follows: Let $H$ be a parity-check matrix of the code $C$. If $C$ is an $[n,k,d]_q$ code, then $H$ is an $(n-k) \times n$ matrix of rank $n-k$ over $\F_q$. Denote by $\mb{h}_1,\ldots,\mb{h}_n$ the $n$ columns of $H$, then the covering radius $\rho(C)$ is the least positive integer $r$ such that for any $\mb{v} \in \F_q^{n-k}$, there is an index set $I \in \binom{[n]}{r}$ such that $\mb{v} \in \langle H_I \rangle$.


\subsection{Generalized Hamming Weights and Generalized Covering Radii}
Let $C$ be an $[n,k,d]_q$ code. For any positive integer $ 1 \le t \le k$, the $t$-th \emph{generalized Hamming weight} (GHW) $d_t(C)$ is defined as the minimum support size of a $t$-dimensional subcode of $C$:
\[
d_t(C) = \min \left\{ \left|\supp(D)\right| :  D \in \sqbinom{C}{t} \right\}. 
\]
It is easy to see that $d_1(C) = d(C)$.

The concept of GHW of linear codes was introduced by
Helleseth et al. \cite{Helleseths} and Kl{\o}ve \cite{Klove}, and was first used by Wei \cite{Wei} to fully characterize the
performance of linear codes when used in a wire-tap channel of type II \cite{Ozarow}. Since then, it has become fundamental parameters of linear codes that have found many applications (see, for example \cite{Forney,Kasami,Hellesethr,Gopalan,Guruswami,Janwa}). 

Similarly, the \emph{generalized covering radii} (GCR) of a code extend the classical covering radius to $t$-dimensional subspaces of $\mathbb{F}_q^n$ \cite{Elimelechl}. There are four equivalent descriptions of the $t$-th GCR parameter $\rho_t(C)$ for a linear code $C$. We use the following two equivalent definitions. First, the $t$-th GCR $\rho_t(C)$ is defined as (see \cite[Definition 6, pp 8073]{Elimelechl})
	\begin{eqnarray} \label{2:rho} \rho_t(C)=\max_{\left(\mb{v}_1,\cdots,\mb{v}_t\right) \in \left(\F_{q}^n\right)^t}\,\,\min_{\left(\mb{c}_1,\cdots,\mb{c}_t\right) \in C^t}\left|\bigcup_{i=1}^t\supp(\mb{v}_i-\mb{c}_i)\right|.\end{eqnarray}
Next, let $H$ be an $(n-k) \times n$ parity-check matrix of $C$. The quantity $\rho_t(C)$ is also the least positive integer $r$ such that for any $t$ distinct vectors $\mb{v}_1,\mb{v}_2,\ldots,\mb{v}_t \in \F_q^{n-k}$, there is an index set $I \in \binom{[n]}{r}$ such that (see \cite[Definition 1, pp 8071]{Elimelechl})
\[\mb{v}_i\in \langle H_I \rangle, \quad \forall \, 1 \le i \le t.\]
The sequence $\rho_1(C), \rho_2(C), \dots$ forms a hierarchy, with $\rho_1(C) = \rho(C)$ recovering the classical covering radius. 


\subsection{Primitive BCH Codes}

For the case that is most interesting to us, we shall focus on a very special subclass of BCH codes, namely the binary (narrow-sense) primitive BCH codes. We define this class of codes below. 

Let $m \ge 2$ and let $\alpha$ be a primitive element of $\mathbb{F}_{2^m}$. For each positive integer $i$, denote by $m_i(x) \in \F_2[x]$ the minimal polynomial of $\alpha^i$ over $\F_2$. Let $e \ge 2$ be an integer. Define 
\[g_{e}(x) \triangleq \mathrm{LCM} \left(m_1(x),m_2(x),\cdots,m_{2e-1}(x)\right),\]
where $\mathrm{LCM}$ denotes the least common multiple of the polynomials over $\F_2$. The binary (narrow-sense) primitive $e$-error-correcting BCH code of length $n =2^m-1$ over $\F_2$, denoted by $\BCH(e,m)$, is the cyclic code of length $n$ with generator polynomial $g_{e}(x)$ over $\F_2$, i.e.,
\[\BCH(e,m)=\left(g_e(x)\right) \subseteq \F_2[x]/(x^n-1).\]
Equivalently, let $H(e,m)$ be the $e \times n$ matrix over $\F_{2^m}$ given by
\begin{align*} \label{2:pcbch}
H(e,m)=
\begin{bmatrix}
1 & \alpha & \alpha^2 & \alpha^3 & \cdots &\alpha^{n-1}\\
1 & \alpha^3 & \alpha^6 & \alpha^9 & \cdots &\alpha^{3(n-1)}\\
\vdots & \vdots & \vdots & \vdots & \vdots &\vdots\\
1 & \alpha^{2e-1} & \alpha^{2(2e-1)} & \alpha^{3(2e-1)} & \cdots &\alpha^{(2e-1)(n-1)}
\end{bmatrix},
\end{align*}
and let $C(e,m)$ be the linear code over $\F_{2^m}$ with $H(e,m)$ as the parity-check matrix, i.e., 
\[C(e,m)=\left\{\mb{x}=(x_1,x_2,\cdots,x_n) \in \F_{2^m}: H(e,m) \cdot \mb{x}^T=\mb{0}\right\}.\]
Here $\mb{x}^T$ is the transpose of the vector $\mb{x}$. Then 
\[\BCH(e,m)=C(e,m)|_{\F_2},\]
that is, $\BCH(e,m)$ is the restriction of $C(e,m)$ on $\F_2$. 
We remark that if $2e-1 \le 2^{\lceil m/2\rceil}$, then the code $\BCH(e,m)$ has parameters $[2^m-1,2^m-me-1,2e+1]$ (see \cite{Cohenar}). We shall often require this condition.

By applying an arbitrary $\F_2$-linear isomorphism $\phi: \F_{2^m} \to \F_{2}^{m \times 1}$ to entries of $H(e,m)$, we can obtain an $me\times n$ matrix $\overline{H(e,m)}$ over $\F_2$, which can serve as a parity-check matrix for $\BCH(e,m)$. Hence, the definition of the generalized covering radius of $\BCH(e,m)$ based on the parity-check matrix $\overline{H(e,m)}$ over $\F_2$ can be translated into conditions based on the matrix $H(e,m)$ over $\F_{2^m}$. So we have the following result which we focus on the case that $e=2$  (see also \cite{Yohananovb,Ozbudakh}): 
\begin{lem}[\cite{Yohananovb,Ozbudakh}] \label{2:covbch}
Let $t \ge 1$. The $t$-th generalized covering radius $\rho_t$ of the code $\BCH(2,m)$ is the least positive integer $r$ such that for any $t$ vectors $\colvec{\alpha_1}{\beta_1}, \cdots, \colvec{\alpha_t}{\beta_t}$ in $\F_{2^m}^2$, there exists $r$ values $x_1,\cdots, x_r \in \mathbb{F}_{2^m}^*$ such that 
	\[    \colvec{\alpha_i}{\beta_i} \in \spa_{\F_2} \left\{\colvec{x_1}{x_1^3}, \cdots, \colvec{x_r}{x_r^3}\right\}, \quad \forall \, 1 \le i \le t. \]\end{lem}

\subsection{Additive Character Sums and Quadratic Equations}

An additive character of $\gf_{q}$ is a function $\psi: \gf_{q}\to\mathbb{C}\setminus \{0\}$ such that 
\[\psi(\beta_1+\beta_2)=\psi(\beta_1)\psi(\beta_2), \quad \forall \beta_1,\beta_2 \in \gf_q.\]
Suppose $\gf_q$ is of characteristic $p$. The canonical additive character is given by
\[\psi(\beta)=\zeta_p^{\mathrm{Tr}_{\gf_q/\gf_p}\left(\beta\right)}, \quad \forall \beta \in \gf_q, \]
where $\zeta_p:=\exp\left(2 \pi \sqrt{-1}/p\right)$ is the complex primitive $p$-th root of unity. When $p=2$, then $\zeta_2=-1$, this character takes values in $\{\pm 1\}$, a fact frequently used in coding theory. 

The following Weil-type bound is the special case $M=0$ of a result of Cochrane and Pinner \cite[Theorem 1.1]{Cochraned}.

\begin{lem}[\cite{Cochraned}]\label{rational}
Let $\gf_q$ be a finite field of characteristic $p$, and let $\psi$ be a nontrivial additive character of $\gf_q$. Suppose that $f(X)\in\gf_q(X)$ is a nonconstant rational function of the form
\[
f(X)=c+\frac{r(X)}{q(X)},
\]
where $c\in\gf_q$, $r(X),q(X)\in\gf_q[X]$, $\deg r<\deg q$, and $\gcd(r,q)=1$. Write
\[
q(X)=\prod_{i=1}^Q q_i(X)^{m_i},
\]
where the $q_i(X)$ are distinct irreducible polynomials over $\gf_q$ and $p\nmid m_i$ for every $i$. Set
\[
L=\sum_{i=1}^Q(m_i+1)\deg q_i(X).
\]
Then
\[
\left|
\sum_{\beta\in\gf_q\setminus S}\psi(f(\beta))
\right|
\leq 1+(L-2)\sqrt q,
\]
where $S$ is the set of poles of $f(X)$ in $\gf_q$.
\end{lem}

We will also use the following well-known criteria for quadratic polynomials over $\F_{2^m}$. 
\begin{lem}\cite{Berlekampe}\label{rootinfq}
	For any $n \in \N$, the quadratic equation $x^2+ax+b=0$, $a,b\in\gf_{2^n}$, $a\neq 0$ has solutions $x \in \gf_{2^n}$ if and only if $\mathrm{Tr}_{\gf_{2^n}/\gf_2}\left({b}/{a^2}\right)=0$. 
\end{lem}

\section{Lower bounds of $\rho_3$ and $\rho_4$}\label{sec:rho3-rho4}

\subsection{The Generalized Supercode Lemma}

The supercode lemma \cite{Cohenas} provides a natural lower bound of the classical covering radius of a code by the minimum distance of its ``supercode''. It turns out the situation is the same for GCR, as can be seen from the following result. 

\begin{lem}[The Generalized Supercode Lemma]\label{lem:supercode}
	Let $C$ and $C'$ be linear codes over $\mathbb{F}_q$ such that $C \subset C' \subseteq \mathbb{F}_q^n$. For any $t \in \N$, if $\dim(C') - \dim(C) \ge t$, then 
	\[
	\rho_t(C) \ge d_t(C, C') \ge d_t(C'). 
	\]
Here $d_t(C')$ is the $t$-th GHW of $C'$, $\rho_t(C)$ is the $t$-th GCR of $C$, and $d_t(C, C')$ is defined as 
\begin{eqnarray*} \label{2:drcc} 
d_t(C, C') = \max_{\left(\mb{c}_1', \dots, \mb{c}_t'\right)  \in \mathcal{T}_t} \,\,\,\, \min_{\left(\mb{c}_1, \dots, \mb{c}_t\right) \in C^t} \left| \bigcup_{i=1}^t \text{supp}\left(\mb{c}_i' - \mb{c}_i\right) \right|,
\end{eqnarray*}
where $\mathcal{T}_t$ is the set of all $t$-tuples $\left(\mb{c}_1',\cdots,\mb{c}_t'\right) \in (C')^t$  that are linearly independent modulo $C$. 
\end{lem}

\begin{proof}
The first inequality $\rho_t(C) \ge d_t(C, C')$ is obvious by the definition of $\rho_t(C)$ given in \eqref{2:rho}. As for the second inequality 
$d_t(C, C') \ge d_t(C')$, we fix an element $\left(\mb{c}_1',\ldots,\mb{c}_t'\right) \in \mathcal{T}_t$ and an element $\left(\mb{c}_1,\ldots,\mb{c}_t\right) \in C^t$. Let 
	\[D=\spa_{\F_q}\{\mb{c}_1'-\mb{c}_1,\ldots,\mb{c}_t'-\mb{c}_t\}.\]
By definition, $D \in \sqbinom{C'}{t}$ and thus 
\[ \left|\bigcup_{i=1}^t\supp(\mb{c}_i'-\mb{c}_i)\right|=\left|\supp(D)\right| \ge d_t(C'). \]
It follows that $d_t(C,C') \ge d_t(C')$, as desired.
\end{proof}

Since $\BCH(1,m)$ is a Hamming code of length $n=2^m-1$, whose GHW are all well understood (\cite[Corollary 4]{Wei}), and $\BCH(2,m) \subset \BCH(1,m)$, and $\dim_{\F_2} \BCH(1,m)-\dim_{\F_2} \BCH(2,m)=m$, by Lemma \ref{lem:supercode}, we immediately obtain lower bounds of the GCR of $\BCH(2,m)$ as below: 
\begin{thm} \label{2:bch1} For any $t \in \N$, if $m \ge \max\{3,t\}$, then we have 
	\[\rho_t\left(\BCH(2,m)\right) \ge d_t,\]
	where $d_1<d_2<d_3< \cdots$ are the GHW of $\BCH(1,m)$ given by 
\[\{d_i:1 \le i \le 2^m-1-m\}=\{1,2,3,\ldots,2^m-1\}\setminus\{2^i:0 \le i <m\}.\] 
In particular, $d_1=3,d_2=5,d_3=6,d_4=7$, so if say $m \ge 4$, we have 
	\begin{align} \label{1:lowerbound} \rho_1(\BCH(2,m)) \ge 3, \quad \rho_2(\BCH(2,m)) \ge 5, \nonumber \\
\rho_3(\BCH(2,m)) \ge 6, \quad \rho_4(\BCH(2,m)) \ge 7.  \end{align} 
\end{thm}
We remark that the lower bound $\rho_2(\BCH(2,m)) \ge 5$ for $m \ge 3$ was originally established by a more complex combinatorial argument (see \cite[Theorem 3]{Yohananovb} for $e=2$ and also \cite[Theorem 4.1]{Ozbudakh})

\subsection{On $\rho_3(\BCH(2,m))$}

In \cite{Ozbudakh}, based on quite skillful and sophisticated combinatorial arguments, occupying 16 pages, \"{O}zbudak and \"{O}zt\"{u}rk provided a detailed study of the lower bound of $\rho_3(\BCH(2,m))$ and proved that (see \cite[Theorems 6.1 and 6.2]{Ozbudakh}) 
\begin{align*} \label{3:ozb} 
	\left\{
	\begin{array}{cl}
		\rho_3(\BCH(2,m)) \ge 6, & \mbox{ if } m \ge 5 \mbox{ is odd},\\
		\rho_3(\BCH(2,m)) \ge 7, & \mbox{ if } m \ge 4 \mbox{ is even}.
	\end{array}
	\right.
\end{align*}
Armed with Lemma \ref{lem:supercode}, we provide a streamlined proof.

We first need some technical results. 

\begin{lem}\label{lem:y1y2y3}
	Let $y_1,y_2,y_3,\alpha \in \F_{2^m}$. 
	Suppose
	$$
	\begin{aligned}
		& y_1+y_2+y_3=0, \\
		& y_1^3+y_2^3+y_3^3=\alpha. 
	\end{aligned}
	$$
	Then $\alpha=y_1y_2(y_1+y_2)$.
\end{lem}
\begin{proof}
	Since $y_3=y_1+y_2$, we have
	\[
	\alpha=y_1^3+y_2^3+y_3^3=y_1^3+y_2^3+(y_1+y_2)^3=y_1^2 y_2+y_1 y_2^2=y_1y_2(y_1+y_2).\qedhere
	\]
\end{proof}
The following lemma has appeared in \cite[Proposition 6.1]{Ozbudakh}. Here we present a simpler alternative proof. 
\begin{lem}\label{lem:cube}
	Let $m\geq 4$ be even. Let $\alpha_1, \alpha_2, \alpha_3 \in \mathbb{F}_{2^m}$ such that $\left\{\alpha_1, \alpha_2, \alpha_3\right\}$ is $\mathbb{F}_2$-linearly independent. 
	If the system
	$$
	\begin{aligned}
		& y_1y_2(y_1+y_2)=\alpha_1, \\
		& y_2y_3(y_2+y_3)=\alpha_2, \\
		& y_3y_1(y_3+y_1)=\alpha_3
	\end{aligned}
	$$
	is solvable with $\left(y_1, y_2, y_3\right) \in \mathbb{F}_{2^m}^3$, then 
	$$
	A=\alpha_1 \alpha_2 \alpha_3\left(\alpha_1+\alpha_2\right)\left(\alpha_1+\alpha_3\right)\left(\alpha_2+\alpha_3\right)\left(\alpha_1+\alpha_2+\alpha_3\right)
	$$
	is a cube in $\F_{2^m}$.
\end{lem}
\begin{proof}
	Note that $$\alpha_1+\alpha_2=y_2(y_1^2+y_1y_2+y_3^2+y_2y_3)=y_2(y_1+y_2+y_3)(y_1+y_3).$$ A similar computation shows that 
	$$
	\begin{aligned}
		& \alpha_2+\alpha_3=y_3(y_1+y_2+y_3)(y_1+y_2), \\
		& \alpha_3+\alpha_1=y_1(y_1+y_2+y_3)(y_2+y_3), \\
		& \alpha_1+\alpha_2+\alpha_3=
		(y_1+y_2)(y_2+y_3)(y_3+y_1).
	\end{aligned}
	$$
	Thus,
	$$
	A=\big(y_1y_2y_3(y_1+y_2)(y_2+y_3)(y_3+y_1)(y_1+y_2+y_3)\big)^3,
	$$
	as required.
\end{proof}

\begin{lem}\label{lem:a1a2a3}
	Let $m\geq 4$ be even. Then there exist $\alpha_1, \alpha_2, \alpha_3 \in \mathbb{F}_{2^m}$ such that $\left\{\alpha_1, \alpha_2, \alpha_3\right\}$ is $\mathbb{F}_2$-linearly independent and
	$$
	A=\alpha_1 \alpha_2 \alpha_3\left(\alpha_1+\alpha_2\right)\left(\alpha_1+\alpha_3\right)\left(\alpha_2+\alpha_3\right)\left(\alpha_1+\alpha_2+\alpha_3\right)
	$$   
	is not a cube in $\F_{2^m}$.
\end{lem}
\begin{proof}
Choose any three $\F_2$-linearly independent elements
$\alpha_1,\alpha_2,\alpha_3\in\F_{2^m}$, and set
\[
V=\spa_{\F_2}\{\alpha_1,\alpha_2,\alpha_3\},
\qquad
A=\prod_{\alpha\in V\setminus\{0\}}\alpha.
\]
Since $m$ is even, there exists a noncube
$\lambda\in\F_{2^m}^*$. If $A$ is not a cube, we are done.
Otherwise, let $\beta_i=\lambda\alpha_i$ for $i=1,2,3$.
These elements are again $\F_2$-linearly independent, and
\[
\begin{aligned}
B&=\prod_{\beta\in
\spa_{\F_2}\{\beta_1,\beta_2,\beta_3\}\setminus\{0\}}\beta=\lambda^7 A.
\end{aligned}
\]
As $\lambda^7=\lambda(\lambda^2)^3$ is a noncube and $A$ is a cube,
$B$ is not a cube, as required.
\end{proof}

Now we can state and prove the lower bound for $\rho_3(\BCH(2,m))$. 

\begin{thm} \label{3:rho3} 

\begin{itemize}
	\item[(i)] If $m\geq 3$, then $\rho_3(\BCH(2,m)) \geq 6$.
	\item[(ii)] If $m\geq 4$ is even, then $\rho_3(\BCH(2,m))\geq 7$.
\end{itemize}
\end{thm}

\begin{proof}
(i) was already known from (\ref{1:lowerbound}), so we only need to prove (ii). Let us assume that $m \ge 4$ is even. Denote $C=\BCH(2,m), C'=\BCH(1,m)$ and $n=2^m-1$, by Lemma \ref{lem:supercode}, it suffices to show that $d_3(C,C') >d_3(C')=6$. 

Suppose on the contrary that $d_3(C,C')=6$. This means that, by the definition of $d_3(C,C')$ and since $d_3(C')=6$, for any $\mb{c}_1',\mb{c}_2',\mb{c}_3' \in C'$ which are linearly independent in $C'/C$, there always exist $\mb{c}_1,\mb{c}_2,\mb{c}_3 \in C$ such that  
\begin{align} \label{3:sup6} \left|\bigcup_{i=1}^3\supp(\mb{c}_i'-\mb{c}_i)\right|=6.\end{align}
Now we choose $\alpha_1, \alpha_2, \alpha_3 \in \mathbb{F}_{2^m}$ such that $\left\{\alpha_1, \alpha_2, \alpha_3\right\}$ is $\mathbb{F}_2$-linearly independent and
	$$
	A=\alpha_1 \alpha_2 \alpha_3\left(\alpha_1+\alpha_2\right)\left(\alpha_1+\alpha_3\right)\left(\alpha_2+\alpha_3\right)\left(\alpha_1+\alpha_2+\alpha_3\right)
	$$   
	is not a cube in $\F_{2^m}$. As $m \ge 4$, such $\alpha_i$'s exist in $\F_{2^m}$ by Lemma~\ref{lem:a1a2a3}. Take $\mb{c}_i' \in \F_2^n$ such that
	\begin{align} \label{3:hvi0}
		H(2,m) \cdot (\mb{c}_i')^T=\colvec{0}{\alpha_i}, \quad 1 \le i \le 3. 
	\end{align} 
	Since $\alpha_i \ne 0$ and $\left\{\alpha_1, \alpha_2, \alpha_3\right\}$ is $\mathbb{F}_2$-linearly independent, we have that $\mb{c}_i' \in C' \setminus C$, and $\mb{c}_1',\mb{c}_2',\mb{c}_3'$ are linearly independent in $C'/C$.  So there exists $\mb{c}_1,\mb{c}_2,\mb{c}_3 \in C$ such that \eqref{3:sup6} holds. 
	Denote $\mb{v}_i=\mb{c}_i'-\mb{c}_i$ for each $i$. Then we have 
	\begin{align} \label{3:hvi}
		H(2,m) \cdot \mb{v}_i^T=\colvec{0}{\alpha_i}, \forall 1 \le i \le 3, \mbox{ and } \left|\bigcup_{i=1}^3\supp(\mb{v}_i)\right|=6.
	\end{align} 
Now, define a binary linear code $E$ given by
\[E =\spa_{\F_2} \left\{\mb{v}_1,\mb{v}_2,\mb{v}_3\right\}, \]
and define \[V=\spa_{\F_2} \left\{\alpha_1,\alpha_2,\alpha_3\right\}.\]
We see that 
\[\left|\supp(E)\right|=\left|\bigcup_{i=1}^3\supp(\mb{v}_i)\right|=6\]
and 
$$
	\forall \mb{v} \in E \setminus \{\mb{0}\}, \,\, H(2,m) \cdot \mb{v}^T=\colvec{0}{\alpha} \Longleftrightarrow \alpha \in V \setminus \{0\}. 
$$
Noting that $\dim_{\F_2} E=3$ and $E \subset C'=\BCH(1,m)$, whose minimum distance is 3, so $\wt(\mb{v}) \ge 3$ for any $\mb{v} \in E \setminus \left\{\mb{0}\right\}$, we conclude that the binary linear code $E$ considered on the support set $I=\supp(E)$ is a binary $[6,3,3]$ code.  	

Since all binary $[6,3,3]$ codes are unique up to permutation equivalence (this result is probably well-known; however, since we can not locate it in the literature, we will prove it in {\bf Appendix A}), and one convenient representative is the shortened Hamming code $H_3$, whose generator matrix is given below: 
	\begin{align} \label{4:gmx}
	G=\begin{bmatrix}
		1&0&0&1&1&0\\
		0&1&0&1&0&1\\
		0&0&1&0&1&1
	\end{bmatrix}.
	\end{align}

Perform the permutation equivalence to transform $E$ to the code $H_3$, which amounts to performing the same permutation equivalence on $\BCH(2,m)$ and $\BCH(1,m)$ (or equivalently permutes the columns of the parity-check matrix $H(2,m)$), then corresponding to the three rows of the matrix $G$ in \eqref{4:gmx}, according to \eqref{3:hvi0} or \eqref{3:hvi}, we can find $\beta_1,\beta_2,\beta_3 \in V$ and $y_1,\ldots,y_6 \in \F_{2^m}$ such that
	$$
	\begin{bmatrix}
		y_1&y_2&y_3&y_4&y_5&y_6\\
		y_1^3&y_2^3&y_3^3&y_4^3&y_5^3&y_6^3
	\end{bmatrix} \cdot G^T =
\begin{bmatrix}
		0&0&0\\
		\beta_1&\beta_2&\beta_3
	\end{bmatrix}.
	$$
Here $G^T$ is the transpose of the matrix $G$. So we have 
	$$\left\{
	\begin{aligned}
		& y_1+y_4+y_5=0, \quad 
		& y_1^3+y_4^3+y_5^3=\beta_1,\\
		& y_2+y_4+y_6=0, \quad 
		& y_2^3+y_4^3+y_6^3=\beta_2,\\
		& y_3+y_5+y_6=0, \quad 
		& y_3^3+y_5^3+y_6^3=\beta_3. 
	\end{aligned}\right.
	$$
By Lemma~\ref{lem:y1y2y3}, we have
	$$
	\left\{\begin{aligned}
		& y_4y_5(y_4+y_5)=\beta_1, \\
		& y_4y_6(y_4+y_6)=\beta_2, \\
		& y_5y_6(y_5+y_6)=\beta_3.
	\end{aligned} \right. 
	$$
	Observing that $\left\{\beta_1,\beta_2,\beta_3\right\}$ is also an $\F_2$-basis of $V$, we also have
	$$
	A=\prod_{\alpha \in V\setminus \{0\}} \alpha=\beta_1\beta_2\beta_3(\beta_1+\beta_2)(\beta_1+\beta_3)(\beta_2+\beta_3)(\beta_1+\beta_2+\beta_3).
	$$
	It follows from Lemma~\ref{lem:cube} that $A$ is a cube in $\F_{2^m}$, violating the assumption that $A$ is not a cube in $\F_{2^m}$. Thus, $d_3(C,C')>6$ and hence $\rho_3(\BCH(2,m)) \geq 7$.
\end{proof}

\subsection{On $\rho_4(\BCH(2,m))$}

While Theorem \ref{2:bch1} implies directly $\rho_4(\BCH(2,m)) \ge 7$ if $m \ge 4$, we can strengthen the result slightly by using a similar argument as in the proof of Theorem \ref{3:rho3}.

\begin{thm}
If $m \ge 4$, then $\rho_4(\BCH(2,m)) \ge 8$. 
\end{thm}
\begin{proof} 

Denote $C=\BCH(2,m), C'=\BCH(1,m)$ and $n=2^m-1$, by Lemma \ref{lem:supercode}, it suffices to show that $d_4(C,C') >d_4(C')=7$. 

Suppose on the contrary that $d_4(C,C')=7$. As in the proof of Theorem \ref{3:rho3}, we choose $\alpha_1, \alpha_2, \alpha_3,\alpha_4 \in \mathbb{F}_{2^m}$ such that $\left\{\alpha_1, \alpha_2, \alpha_3,\alpha_4\right\}$ is $\mathbb{F}_2$-linearly independent. Since $m \ge 4$, such $\alpha_i$'s exist in $\F_{2^m}$. Take $\mb{c}_i' \in \F_2^n$ such that
	\begin{align*}
		H(2,m) \cdot (\mb{c}_i')^t=\colvec{0}{\alpha_i}, \quad 1 \le i \le 4, 
	\end{align*} 
and choose $\mb{c}_1,\mb{c}_2,\mb{c}_3,\mb{c}_4 \in C$ such that 
\[\left|\bigcup_{i=1}^4 \supp(\mb{c}_i'-\mb{c}_i)\right|=7.\] 
	Denote $\mb{v}_i=\mb{c}_i'-\mb{c}_i$ for each $i$. Then we have 
	\begin{align*} \label{3:hvi-3}
		H(2,m) \cdot \mb{v}_i^t=\colvec{0}{\alpha_i}, \forall 1 \le i \le 4, \mbox{ and } \left|\bigcup_{i=1}^4\supp(\mb{v}_i)\right|=7.
	\end{align*} 
Now, define a binary linear code $E$ given by
\[E =\spa_{\F_2} \left\{\mb{v}_1,\mb{v}_2,\mb{v}_3,\mb{v}_4\right\}, \]
and define 
\[V=\spa_{\F_2} \left\{\alpha_1,\alpha_2,\alpha_3,\alpha_4\right\}.\]
We see that 
\[\left|\supp(E)\right|=\left|\bigcup_{i=1}^4\supp(\mb{v}_i)\right|=7\]
and 
$$
	\forall \mb{v} \in E \setminus \{\mb{0}\}, \,\, H(2,m) \cdot \mb{v}^t=\colvec{0}{\alpha} \Longleftrightarrow \alpha \in V \setminus \{0\}. 
$$
Noting that $\dim_{\F_2} E=4$ and $E \subset C'=\BCH(1,m)$, whose minimum distance is 3, so $\wt(\mb{v}) \ge 3$ for any $\mb{v} \in E \setminus \left\{\mb{0}\right\}$, we conclude that the binary linear code $E$ considered on the support set $I=\supp(E)$ is a binary $[7,4,3]$ code.  	

It is well-known that all binary $[7,4,3]$ codes are permutation equivalent to the Hamming code $\mathcal{H}_3$ \cite{MacWilliamsd}, with a standard systematic generator matrix given by 
\begin{align*} \label{3:rho4}
G=\begin{bmatrix}
	1&0&0&0&1&1&0\\
	0&1&0&0&1&0&1\\
	0&0&1&0&0&1&1\\
	0&0&0&1&1&1&1
\end{bmatrix}.
\end{align*}
The four linearly independent rows of $G$ correspond to $\beta_1,\beta_2,\beta_3,\beta_4 \in V$ which are $\F_2$-linearly independent and satisfy the equations 
$$\left\{
\begin{aligned}
	& y_1+y_5+y_6=0, \quad 
	& y_1^3+y_5^3+y_6^3=\beta_1,\\
	& y_2+y_5+y_7=0, \quad 
	& y_2^3+y_5^3+y_7^3=\beta_2,\\
	& y_3+y_6+y_7=0, \quad 
	& y_3^3+y_6^3+y_7^3=\beta_3,\\
	& y_4+y_5+y_6+y_7=0, \quad 
	& y_4^3+y_5^3+y_6^3+y_7^3=\beta_4,
\end{aligned}\right.
$$
for some $y_1,\ldots,y_7 \in \F_{2^m}$. For the first 3 equations, by Lemma~\ref{lem:y1y2y3}, we have
$$\left\{
\begin{aligned}
	& y_5y_6(y_5+y_6)=\beta_1, \\
	& y_5y_7(y_5+y_7)=\beta_2,\\
	& y_6y_7(y_6+y_7)=\beta_3.
\end{aligned}\right.
$$
For the last equation, since $y_4=y_5+y_6+y_7$, we have
\begin{align*}
\beta_4&=(y_5+y_6+y_7)^3+y_5^3+y_6^3+y_7^3\\
&=y_5y_6(y_5+y_6)+y_5y_7(y_5+y_7)+y_6y_7(y_6+y_7)=\beta_1+\beta_2+\beta_3, 
\end{align*}
contradicting the assumption that $\{\beta_1,\beta_2,\beta_3,\beta_4\}$ is $\F_2$-linearly independent. So, $\rho_4 (\BCH(2,m)) \geq 8$.
\end{proof}

\section{General bounds on $\rho_t(\BCH(2,m))$}\label{sec:general-bounds}

Having established exact lower bounds for small values of $t$ using the Generalized Supercode Lemma, we now turn our attention to the asymptotic behavior of the GCR hierarchy for general $t$ when $m$ is large. 

\subsection{Lower bound for large $m$}
It was known that the GHW parameters of $\BCH(1,m)$ satisfies: if $d_t(\BCH(1,m))=s$, then $t=s-\lceil \log_2 (s)\rceil$ where $s$ is a positive integer which is not a power of 2 (see \cite{Wei}). It can be shown that 
\[t+\lfloor \log_2 t \rfloor+1 \le d_t(\BCH(1,m)) \le t+\lfloor \log_2 t \rfloor+2.\]
Then Theorem~\ref{2:bch1} implies that $\rho_t\left(\BCH(2,m)\right) \geq t+1+\lfloor \log_2 t\rfloor$ for all $m\geq \max\{3,t\}$. The following theorem shows that this lower bound can be significantly improved to $2t$ if $m\geq t(2t-1)$. 

\begin{thm} \label{7:thm1} For any $t \in \N$, if 
	\[2^m \ge \frac{2^{t(2t-1)}}{(2t-1)!},\]
	then $\rho_t\left(\BCH(2,m)\right) \geq 2t$. 
\end{thm}
\begin{proof}
	For each $x \in \F_{2^m}$, denote $\mb{v}(x)=\colvec{x}{x^3}$. 
    
    Suppose $\rho_t(\BCH(2,m)) \le 2t-1$, then for any $(\mb{a_1},\cdots,\mb{a_t}) \in \left(\F_{2^m}^2\right)^t$, there exist $x_1,\cdots,x_{2t-1} \in \F_{2^m}$ such that 
	\[(\mb{a_1},\cdots,\mb{a_t}) \in \left(\spa_{\F_2} \left\{\mb{v}(x_1), \cdots, \mb{v}(x_{2t-1})\right\}\right)^t,\]
	so we have 
	\[\left(\F_{2^m}^2\right)^t \subseteq \bigcup_{\{x_1,\cdots,x_{2t-1}\} \subseteq \F_{2^m}} \left(\spa_{\F_2} \left\{\mb{v}(x_1), \cdots, \mb{v}(x_{2t-1})\right\}\right)^t.\]
	By a simple counting argument, we obtain 
	\[(2^{2m})^{t} \le \binom{2^m}{2t-1} \left(2^{2t-1}\right)^t < \frac{(2^m)^{2t-1}}{(2t-1)!}2^{t(2t-1)},\]
	which implies immediately that 
	\[2^m < \frac{2^{t(2t-1)}}{(2t-1)!}.\]
This completes the proof of Theorem \ref{7:thm1}. 
\end{proof}
\subsection{Upper bound for large $m$}

\begin{thm} \label{7:thm2}
	For any $t \ge 2$, if $m \ge 2t+3+2 \log_2(t-1)$, then $\rho_t \left(\BCH(2,m)\right)\leq 2t+1$. 
\end{thm}

\begin{proof}
	For any $t$ vectors $\colvec{a_1}{b_1}, \cdots, \colvec{a_t}{b_t} \in \F_{2^m}^2$, we shall prove that there exist $x,y_1,\cdots,y_t,z_1,\cdots,z_t \in \F_{2^m}$ such that 
	\[\left\{ \colvec{a_1}{b_1},\cdots,\colvec{a_t}{b_t}\right\} \subseteq \spa_{\F_2}\left\{\colvec{x}{x^3},\colvec{y_i}{y_i^3}, \colvec{z_{i}}{z_{i}^3}: 1 \le i \le t\right\}.\]	Actually we will prove a more explicit result. For these vectors $\colvec{a_i}{b_i}$, consider the system of equations $(1 \le i \le t)$
	\begin{eqnarray} \label{7:ab}
		\colvec{a_i}{b_i}=\colvec{x}{x^3}+\colvec{y_i}{y_i^3}+\colvec{z_i}{z_i^3} \Longleftrightarrow \left\{\begin{array}{ccc} a_i&=&x+y_i+z_i,\\
			b_i&=&x^3+y_i^3+z_i^3. 
		\end{array} \right.  
	\end{eqnarray}
	Denote by $N$ the number of solutions $(x,y_1,\cdots,y_t,z_1,\cdots,z_t) \in \F_{2^m}^{2t+1}$ that satisfy \eqref{7:ab} simultaneously for all $1 \le i \le t$. We will prove that $N>0$ for any fixed $a_i,b_i \in \F_{2^m}$. This would imply that $\rho_t \left(\BCH(2,m)\right)\leq 2t+1$.

	Now we consider \eqref{7:ab}. We can eliminate the variable $z_i$ from \eqref{7:ab} as follows: since $z_i=a_i+x+y_i$, we obtain 
	$b_i=x^3+y_i^3+(a_i+x+y_i)^3$, which implies that 
	\begin{eqnarray} \label{7:eqbi}
		(a_i+x)y_i^2+(a_i^2+x^2) y_i+a_ix^2+a_i^2x+a_i^3+b_i=0.     
	\end{eqnarray}
	For a given $x$ such that $x \ne a_i$, by Lemma \ref{rootinfq}, Eq \eqref{7:eqbi} is solvable for $y_i \in \F_{2^m}$ (with exactly two distinct solutions) if and only if 
	\begin{eqnarray*} 
	0&=&\tr_{\F_{2^m}/\F_2}\left(\frac{a_ix^2+a_i^2x+a_i^3+b_i}{(a_i+x)^3}\right)\\
	&=&\tr_{\F_{2^m}/\F_2}\left(\frac{a_i}{a_i+x}+\frac{a_i^2}{(a_i+x)^2}+\frac{a_i^3+b_i}{(a_i+x)^3}\right)\\
	&=&\tr_{\F_{2^m}/\F_2}\left(\frac{a_i^3+b_i}{(a_i+x)^3}\right), \end{eqnarray*}
	or equivalently, the number of solutions $y_i \in \F_{2^m}$ is given by
	\[1+\psi\left(\frac{a_i^3+b_i}{(a_i+x)^3}\right),\]
	where $\psi: \F_{2^m} \to \{-1,1\}$ given by $x \mapsto (-1)^{\tr_{\F_{2^m}/\F_2}(x)}$ for any $x \in \F_{2^m}$ is the standard additive character on $\F_{2^m}$. 
	From this, we have
	\begin{eqnarray} \label{7:nvau}
		N \ge \sum_{x \in \F_{2^m} \setminus A} \prod_{i=1}^t \left\{1+\psi\left(\frac{a_i^3+b_i}{(a_i+x)^3}\right)\right\}.     
	\end{eqnarray}
	where $A$ is the finite set given by  
	\[A \triangleq \{a_i: 1 \le i \le t\}.\] 
We may expand the right-hand side of \eqref{7:nvau} as 
	\begin{eqnarray} \label{7:sums1} \sum_{x \in \F_{2^m}\setminus A} \left(1+\sum_{\emptyset \ne I \subseteq [t]} \psi\left(f_I(x)\right)\right),\end{eqnarray}
	where for each nonempty subset $I \subseteq [t]$, $f_I(X) \in \F_{2^m}(X)$ is the rational function defined by
	\[f_I(X)=\sum_{i \in I} \frac{a_i^3+b_i}{(a_i+X)^3}. \]
	It is possible that $f_I(X) = 0$ (a zero polynomial) (for example, if $I=\{1,2,3,4\}$ and $a_1=a_2=a_3=a_4$ and $b_1+b_2+b_3+b_4=0$ then $f_I(X) = 0$, even if $a_1^3+b_1 \ne 0$). So Eq \eqref{7:sums1} can be further simplified as 
	\begin{eqnarray} \label{7:sums2} 
\sum_{x \in \F_{2^m}\setminus A} \left(1+\sum_{\substack{\emptyset \neq I \subseteq [t]\\
				f_I(X) \ne 0}} \psi\left(f_I(x)\right)+\sum_{\substack{\emptyset \neq I \subseteq [t]\\
				f_I(X) = 0}} 1\right) \nonumber \\
\ge \sum_{x \in \F_{2^m}\setminus A} \left(1+\sum_{\substack{\emptyset \neq I \subseteq [t]\\
				f_I(X) \ne 0}} \psi\left(f_I(x)\right)\right).\end{eqnarray}
	For each subset $I$, denote 
	\[A_I \triangleq \{a_i: i \in I\}. \]
	The right-hand side of \eqref{7:sums2} can be written as 
	\begin{eqnarray*} \label{7:sums3}
		2^m-|A|+\sum_{\substack{\emptyset \neq I \subseteq [t]\\
				f_I(X) \ne 0}} \left(\sum_{x \in \F_{2^m} \setminus A_I} \psi\left(f_I(x)\right)-\sum_{x \in A \setminus A_I} \psi\left(f_I(x)\right)\right).
	\end{eqnarray*}
	Trivially we have
	\[\sum_{\substack{\emptyset \neq I \subseteq [t]\\
			f_I(X) \ne 0}} \left|\sum_{x \in A \setminus A_I} \psi\left(f_I(x)\right)\right| \le \left(2^t-1\right)(|A|-1).\]
Moreover, for each nonempty subset $I\subseteq [t]$ such that $f_I(X)\neq 0$, let $S_I\subseteq A_I$ be the set of poles of $f_I(X)$. Each pole has order $3$, so Lemma \ref{rational}, applied with $p=2$ and $L=4|S_I|$, gives
	\[
	\begin{aligned}
	\left|\sum_{x\in\F_{2^m}\setminus A_I}\psi(f_I(x))\right|
	&\le 1+(4|S_I|-2)\sqrt{2^m}+|A_I\setminus S_I|\\
	&\le 1+(4|I|-2)\sqrt{2^m}.
	\end{aligned}
	\]
	Combining all the estimates above, we have 
	\begin{eqnarray*}
		N &\ge& 2^m-|A|-\left(2^t-1\right)\left(|A|-1\right)-\sum_{\emptyset \ne I \subseteq [t]} \left(1+\left(4|I|-2\right)\sqrt{2^m}\right)\\
		&=&2^m-|A|2^t-2\sqrt{2^m}\left((t-1)2^t+1\right)\\
		& \ge &2^m-t2^t-2\sqrt{2^m}\left((t-1)2^t+1\right).
	\end{eqnarray*}
	Now it is easy to see that if 
	\begin{eqnarray} \label{7:condm} \sqrt{2^m} \ge (t-1)2^{t+1}+3,\end{eqnarray} then 
	\[N \ge \sqrt{2^m}\left(\sqrt{2^m}-(t-1)2^{t+1}-2\right)-t2^t \ge \sqrt{2^m}-t2^t>0. \]
	Taking logarithm on both sides of \eqref{7:condm}, it is easy to see that for any $t \ge 2$, if $m \ge 2t+3+2 \log_2(t-1)$, then \eqref{7:condm} holds. This completes the proof of Theorem \ref{7:thm2}.
\end{proof}

\section{Conclusion}\label{sec:cond} 

In this paper, we introduced the Generalized Supercode Lemma, establishing a natural connection between the generalized covering radii (GCR) of a linear code and the generalized Hamming weights of an appropriate supercode. By combining this lemma with coding-theoretic arguments, we significantly streamlined the proofs for the lower bounds of \(\rho_2(\text{BCH}(2,m))\) and \(\rho_3(\text{BCH}(2,m))\), which previously required highly complex combinatorial analyses, and derived new lower bounds for \(\rho_4(\text{BCH}(2,m))\). Furthermore, by utilizing Weil-type exponential sum estimates, we bounded the GCR hierarchy for general orders, proving that \(2t \le \rho_t(\text{BCH}(2,m)) \le 2t+1\) whenever \(m\) is sufficiently large relative to \(t\).

While this work resolves the asymptotic behavior of the GCR hierarchy for these codes, determining the exact value of \(\rho_t(\text{BCH}(2,m))\) for arbitrary \(t\) and smaller \(m\) remains an open problem. Additionally, exploring the application of the Generalized Supercode Lemma to evaluate the GCR of other important families of linear codes presents an exciting direction for future research.

\appendices 
\section{} \label{sec:appenx}

\begin{lem}\label{lem:unique_6_3_3}
The binary linear code with parameters $[6,3,3]$ is unique up to permutation equivalence.
\end{lem}

\begin{proof}
Let $C$ be a binary $[6,3,3]$ linear code. Its parity-check matrix $H$ has dimensions $3 \times 6$. 

Since $d(C)=3$, any two columns of $H$ must be linearly independent. Over $\mathbb{F}_2$, this is equivalent to saying that all $6$ columns of $H$ are non-zero and mutually distinct. 

The space of column vectors $\mathbb{F}_2^3$ contains exactly $2^3 - 1 = 7$ non-zero vectors. Thus, the columns of $H$ comprise all but one of the non-zero vectors in $\mathbb{F}_2^3$. Let $\mb{v} \in \mathbb{F}_2^3 \setminus \{\mb{0}\}$ denote this unique ``missing'' vector.

Now, suppose $C'$ is another binary $[6,3,3]$ code with parity-check matrix $H'$ and corresponding missing vector $\mb{v}' \in \mathbb{F}_2^3 \setminus \{0\}$. 

The general linear group $\mathrm{GL}(3,\F_2)$ (the group of all $3 \times 3$ invertible matrices over $\F_2$) acts transitively on the set of non-zero vectors in $\mathbb{F}_2^3$. Therefore, there exists an invertible matrix $M \in \mathrm{GL}(3,\F_2)$ such that $M\mb{v} = \mb{v}'$. 

Consider the matrix $MH$. Since $M$ is invertible, $MH$ is also a parity-check matrix for the same code $C$. The columns of $MH$ are obtained by applying the bijection $M$ to the columns of $H$. Consequently, the columns of $MH$ consist of all non-zero vectors in $\mathbb{F}_2^3$ except for $M\mb{v} = \mb{v}'$. Thus $MH$ and $H'$ have the exact same set of columns (all non-zero vectors in $\mathbb{F}_2^3$ except $\mb{v}'$), they differ only by a permutation of their columns, i.e., the codes $C$ and $C'$ are permutation equivalent.
\end{proof}

\section*{Acknowledgment}
The first author was supported by the Research Grants Council (RGC) of Hong Kong (No. 16307524). The second author thanks the Hong Kong University of Science and Technology for hospitality during his visit in January 2026, where this project was initiated. The authors are grateful to anonymous referees for their valuable comments and suggestions. 
The authors thank Darrion Thornburgh for suggesting the simpler proof of Lemma~\ref{lem:a1a2a3}.

\bibliographystyle{IEEEtranS}



\end{document}